\documentclass[draftclsnofoot, onecolumn, letterpaper, romanappendices]{IEEEtran}

\usepackage{Albert_Style}

\begin{document}
\title{Universality of Logarithmic Loss in Lossy Compression}

\author{\IEEEauthorblockN{Albert No, \emph{Member, IEEE}},
and 
\IEEEauthorblockN{Tsachy Weissman, \emph{Fellow, IEEE}\\}
\thanks{The material in this paper has been presented in part at the 2015 International
Symposium on Information Theory.
This work was supported by the Hongik University new faculty research support fund.}
\thanks{A.\ No is with the Department of Electronic and Electrical Engineering,
Hongik University, Seoul, Korea (email: albertno@hongik.ac.kr)}
\thanks{T.\ Weissman is with the Department of Electrical Engineering, Stanford University,
Stanford, CA 94305, USA (email: tsachy@stanford.edu)}
}

\maketitle

\begin{abstract}
We establish two strong senses of universality of logarithmic loss as a
distortion criterion in lossy compression: For any fixed length lossy
compression problem under an arbitrary distortion criterion,
we show that there is an equivalent lossy compression problem under logarithmic loss.
In the successive refinement problem, if the first decoder operates under logarithmic loss,
we show that any discrete memoryless source is successively refinable under an
arbitrary distortion criterion for the second decoder.

\end{abstract}
\IEEEpeerreviewmaketitle

\section{Introduction}\label{sec:Introduction}
In the lossy compression problem, logarithmic loss distortion is a criterion
allowing a ``soft" reconstruction of the source, a departure from the classical
setting of a deterministic reconstruction.
Although logarithmic loss plays a crucial role in the theory of learning and prediction,
relatively little work has been done in the context of lossy compression,
notwithstanding the two-encoder multiterminal source coding problem under logarithmic loss
\cite{courtade2011multiterminal, courtade2014multiterminal},
or recent work on the single-shot approach to lossy courec coding under logarithmic loss \cite{shkel2016singleshot}.
Note that lossy compression under logarithmic loss is closely related to the information
bottleneck method \cite{tishby1999bottleneck, harremoes2007information, gilad2003information}.
In this paper, we focus on universal properties of logarithmic loss in two lossy compression problems. 

First, we consider the fixed-length lossy compression problem.
We show that there is a one-to-one correspondence between any fixed-length lossy compression
problem under an arbitrary distortion measure and that under logarithmic loss.
The correspondence is in the following strong sense:
\begin{itemize}
    \item Optimum schemes for the two problems are the same.
    \item A good scheme for one problem is also a good scheme for the other.
\end{itemize}
We will be more precise about ``optimum" and ``goodness" of the scheme in later sections.
This finding essentially implies that it is enough to consider the lossy compression problem
under logarithmic loss.
We point out that our result is different from that of \cite{jiao2015justification}
which justifies  logarithmic loss by showing it is the only loss function
that satisfies a natural data processing requirement.

We also consider the successive refinement problem under logarithmic loss.
Successive refinement is a network lossy compression problem where one encoder wishes
to describe the source to two decoders \cite{equitz1991successive, koshelev1980hierarchical}.
Instead of having two separate coding schemes, the successive refinement encoder designs
a code for the decoder with a weaker link, and sends extra information to the second decoder
on top of the message of the first decoder.
In general, successive refinement coding cannot do as well as
two separate encoding schemes optimized for the respective decoders.
However, if we can achieve the point-to-point optimum rates using successive refinement coding,
we say the source is successively refinable.
We show that any discrete memoryless source is successively refinable as long as the weaker link
employs logarithmic loss, regardless of the distortion criterion used for the stronger link.
This idea can be useful to construct practical point-to-point lossy compression
since it allows a smaller codebook and lower encoding complexity
\cite{no2016strong, venkataramanan2014lossy, no2016rateless}. 

The remainder of the paper is organized as follows.
In Section \ref{sec:Preliminaries}, we revisit some of the known results pertaining to
logarithmic loss and fixed-length lossy compression.
Section \ref{sec:Equivalence} is dedicated to the equivalence between lossy compression under
arbitrary distortion measures and that under logarithmic loss.
Section \ref{sec:Successive Refinability} is dedicated to successive refinement under
logarithmic loss in the weaker link.

\emph{Notation}: $X^n$ denotes an $n$-dimensional random vector $(X_1,X_2,\ldots,X_n)$
while $x^n$ denotes a specific possible realization of the random vector $X^n$.
Similarly, $Q$ denotes a random probability mass function
while $q$ denotes a specific probability mass function.
We use natural logarithm and nats instead of bits.

\section{Preliminaries}\label{sec:Preliminaries}

\subsection{Logarithmic Loss}\label{subsec:Logarithmic Loss}
Let $\cX$ be a set of discrete source symbols ($|\cX|<\infty$),
and $\cM(\cX)$ be the set of probability measures on $\cX$.
Logarithmic loss $\ell:\cX\times \cM(\cX) \ra [0,\infty]$ is defined by
\begin{align*}
    \ell(x,q) = \log \frac{1}{q(x)}
\end{align*}
for $x\in\cX$ and $q\in\cM(\cX)$. Logarithmic loss between $n$-tuples is defined by
\begin{align*}
    \ell_n(x^n,q^n) = \frac{1}{n}\sum_{i=1}^n\log \frac{1}{q_i(x_i)},
\end{align*}
i.e., the symbol-by-symbol extension of the single letter loss.

Let $X^n$ be the discrete memoryless source with distribution $P_X$.
Consider the lossy compression problem under logarithmic loss where
the reconstruction alphabet is $\cM(\cX)$.
The rate-distortion function is given by
\begin{align*}
    R(D) =& \inf_{P_{Q|X}: \E{\ell(X,Q)}\leq D} I(X;Q)\\
         =& H(X)-D.
\end{align*}

The following lemma provides a property of the rate-distortion function achieving
conditional distribution.
\begin{lemma}\label{lem:iff condition}
    The rate-distortion function achieving conditional distribution $P_{Q^\star|X}$ satisfies
    \begin{align}
        P_{X|\Qstar}(\cdot|q) =& q \label{eq:iff condition}\\
        H(X|Q^\star) =& D\label{eq:iff condition2}
    \end{align}
    for $P_{\Qstar}$ almost every $q\in\cM(\cX)$. 
    Conversely, if $P_{Q|X}$ satisfies \eqref{eq:iff condition} and \eqref{eq:iff condition2},
    then it is a rate-distortion function achieving conditional distribution, i.e.,
    \begin{align*}
        I(X;Q) =& R(D) = H(X)-D\\
        \E{\ell(X,Q)} =& D.
    \end{align*}
\end{lemma}
The key idea is that we can replace $Q$ by $P_{X|Q}(\cdot|Q)$,
and have lower rate and distortion, i.e.,
\begin{align*}
    I(X;Q) \geq& I(X;P_{X|Q}(\cdot|Q))\\
    \E{\ell(X,Q)} \geq& \E{\ell(X,P_{X|Q}(\cdot|Q)},
\end{align*}
which directly implies \eqref{eq:iff condition}.

Interestingly, since the rate-distortion function in this case is a straight line,
a simple time sharing scheme achieves the optimal rate-distortion tradeoff.
If the encoder losslessly compresses only the first $\frac{H(X)-D}{H(X)}$ fraction
of the source sequence components, while the decoder perfectly recovers those components
and uses $P_X$ as its reconstruction for the remaining components then the resulting
scheme obviously achieves distortion $D$ with rate $H(X)-D$.

Moreover, this simple scheme directly implies successive refinability of the source.
For $D_1>D_2$, suppose the encoder losslessly compresses the first $\frac{H(X)-D_2}{H(X)}$
fraction of the source.
Then the first decoder can perfectly reconstruct $\frac{H(X)-D_1}{H(X)}$ fraction of the source
with the message of rate $H(X) - D_1$ and distortion $D_1$
while the second decoder can achieve distortion $D_2$ with rate $H(X)-D_2$.
Since both decoders can achieve the best rate-distortion pair,
it follows that any discrete memoryless source under logarithmic loss is successively refinable.

\subsection{Fixed-Length Lossy Compression}\label{subsec:One Shot}
In this section, we briefly introduce the basic settings of the fixed-length lossy compression problem
(see \cite{kostina2012fixed} and references therein for more details).
In the fixed-length lossy compression setting, we have a source $X$ with finite alphabet
$\cX = \{1,\ldots,r\}$ and source distribution $P_X$.
An encoder $f:\cX\rightarrow \{1,\ldots,M\}$ maps the source symbol to one of $M$ messages.
On the other side, a decoder $g:\{1,\ldots,M\}\rightarrow \cXh$ maps the message to the actual
reconstruction $\Xh$ where the reconstruction alphabet is also finite $\cXh = \{1,\ldots,s\}$.
Let $d:\cX\times \cXh\rightarrow[0,\infty)$ be a distortion measure between
the source and the reconstruction.

First, we can think of the code that the expected distortion is lower than a given distortion level.
\begin{definition}[Average distortion criterion]
    An $(M,D)$ code is a pair of an encoder $f$ with $|f|\leq M$ and a decoder $g$ such that
    \begin{align*}
        \E{d(X,g(f(X)))}\leq D.
    \end{align*}
    The minimum number of codewords required to achieve average distortion not exceeding $D$ is defined by
    \begin{align*}
        M^\star(D) = \min\{M: \exists (M,D)\mbox{ code}\}.
    \end{align*}
    Similarly, we can define the minimum achievable average distortion
    given the number of codewords $M$.
    \begin{align*}
        D^\star(M) = \min\{D: \exists (M,D)\mbox{ code}\}.
    \end{align*}
\end{definition}

One may consider a stronger criterion that restricts the probability of exceeding a given distortion level.
\begin{definition}[Excess distortion criterion]
    An $(M,D,\epsilon)$ code is a pair of an encoder $f$ with $|f|\leq M$ and a decoder $g$ such that
    \begin{align*}
        \Pr{d(X,g(f(X)))>D}\leq \epsilon.
    \end{align*}
    
    The minimum number of codewords required to achieve excess distortion probability $\epsilon$
    and distortion $D$ is defined by
    \begin{align*}
        M^\star(D,\epsilon) = \min\{M: \exists (M,D,\epsilon)\mbox{ code}\}.
    \end{align*}
    
    Similarly, we can define the minimum achievable excess distortion probability
    given the target distortion $D$ and the number of codewords $M$.
    \begin{align*}
        \epsilon^\star(M,\epsilon) = \min\{\epsilon: \exists (M,D,\epsilon)\mbox{ code}\}.
    \end{align*}
\end{definition}

Given the target distortion $D$ and $P_X$, the informational rate-distortion function is defined by
\begin{align*}
    R(D) = \inf_{P_{\Xh|X}: \E{d(X,\Xh)}\leq D} I(X;\Xh)
\end{align*}
We make the following benign assumptions:
\begin{itemize}
    \item There exists a unique rate-distortion function achieving
        conditional distribution $P_{\Xh|X}^\star$.
    \item We assume that $P_{\Xh^\star}(\xh)>0$ for all $\xh\in\cXh$
        since we can always discard the reconstruction symbol with zero probability.
    \item If $d(x,\xh_1) = d(x,\xh_2)$ for all $x\in\cX$, then $\xh_1 = \xh_2$.
\end{itemize}

Define the information density of the joint distribution $P_{X,\Xh}$ by
\begin{align*}
    \imath_{X;\Xh}(x;\xh)  = \log \frac {P_{X,\Xh}(x,\xh)}{P_X(x) P_{\Xh}(\xh)}.
\end{align*}
Then, we are ready to define the $D$-tilted information which plays a key role
in fixed-length lossy compression.
\begin{definition}\cite[Definition 6]{kostina2012fixed}
    The $D$-tilted information in $x\in\cX$ is defined as
    \begin{align*}
        \jmath_X(x,D) = \log\frac{1}{\E{\exp\left(\lambda^{\star} D
            - \lambda^{\star} d(x,\Xh^{\star})\right)}}
    \end{align*}
    where the expectation is with respect to the marginal distribution of $\Xh^{\star}$
    and $\lambda^{\star} = -R'(D)$.
\end{definition}

\begin{theorem}\cite[Lemma 1.4]{csiszar1974extremum}
    For all $\xh\in\cXh$, 
    \begin{align}
        \jmath_X(x,D)= \imath_{X;\Xh^{\star}}(x;\xh) + \lambda^{\star} d(x,\xh)
            -\lambda^{\star} D,\label{eq:optimum distribution with tilted information}
    \end{align}
    and therefore we have
    \begin{align*}
        R(D) = \E{\jmath(X,D)}.
    \end{align*}
\end{theorem}
Let $P_{X|\Xh}^\star$ be the induced conditional probability from $P_{\Xh|X}^\star$.
Then \eqref{eq:optimum distribution with tilted information} can equivalently be expressed as
\begin{align}
    &\log \frac{1}{P_{X|\Xh}^\star(x|\xh)}\nonumber\\
    &= \log \frac{1}{P_X(x)} -\jmath_X(x,D) +\lambda^\star d(x,\xh) - \lambda^\star D.
        \label{eq:PXgivenPXhstar}
\end{align}

The following lemma shows that $P^{\star}_{X|\Xh}(\cdot|\xh)$ are all distinct.
\begin{lemma}\label{lem:uniqueness of rows}
    For all $\xh_1\neq \xh_2$, there exists $x\in\cX$ such that
    $P^{\star}_{X|\Xh}(x|\xh_1)\neq P^{\star}_{X|\Xh}(x|\xh_2)$.
\end{lemma} 

\begin{proof}
    Suppose by contradiction that $P_{X|\Xh}^{\star}(x|\xh_1) = P_{X|\Xh}^{\star}(x|\xh_2)$ for all $x\in \cX$.
    Then, \eqref{eq:PXgivenPXhstar} implies
    \begin{align*}
    &\log \frac{1}{P_X(x)} -\jmath_X(x,D) +\lambda^\star d(x,\xh_1) - \lambda^\star D\\
    &=\log \frac{1}{P_X(x)} -\jmath_X(x,D) +\lambda^\star d(x,\xh_2) - \lambda^\star D
    \end{align*}
    for all $x\in\cX$, i.e., $d(x,\xh_1) = d(x,\xh_2)$ for all $x\in\cX$,
    which violates our assumption on the distortion measure.
\end{proof}

\subsection{Successive Refinability}\label{subsec:Successive Refinability}
In this section, we review the successive refinement problem with two decoders.
Let the source $X^n$ be i.i.d.\ random vector with distribution $P_X$.
The encoder wants to describe $X^n$ to two decoders by sending a pair of messages $(m_1,m_2)$
where $1\leq m_i\leq M_i$ for $i \in \{1,2\}$.
The first decoder reconstructs $\hat{X}_1^n(m_1)\in\hat{\cX}_1^n$ based only on the first message $m_1$.
The second decoder reconstructs $\hat{X}_2^n(m_1,m_2)\in\hat{\cX}_2^n$ based on both $m_1$ and $m_2$.
The setting is described in Figure \ref{fig:Successive Refinement}.

\tikzstyle{format} = [thin]
\tikzstyle{medium} = [rectangle, draw, thin, fill=blue!20, minimum height=2.5em, minimum width = 4em]
\begin{figure}[h]
\centering
\begin{tikzpicture}[node distance=2.8cm, auto,>=latex', thick]
    \node [format] (src) {$X^n$};
    \node [medium, right of=src, node distance = 2cm](enc){Enc};
    \node [medium, right of=enc, node distance = 3cm](dec1){Dec 1};
    \node [medium, below of=dec1, node distance = 1.5cm](dec2){Dec 2};
    \node [format, right of=dec1, node distance = 2cm](rec1){$\hat{X}_1^n$};
    \node [format, right of=dec2, node distance = 2cm](rec2){$\hat{X}_2^n$};
    \draw [->] (src) -- node {}(enc);
    \draw [->] (enc) -- node {$m_1$}(dec1);
   \draw  [->] (enc.east) -- ++(.7,0)  -- ++(0,-1.3)-- ++(.87,0);
    \draw [->] (enc) |- node [above right] {$m_2$}(dec2);
    \draw [->] (dec1) -- node{} (rec1);
    \draw [->] (dec2) -- node{} (rec2);
\end{tikzpicture}
\caption{Successive Refinement}\label{fig:Successive Refinement}
\end{figure}
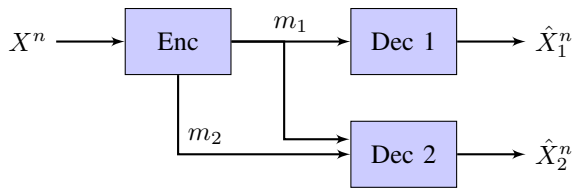 
Let $d_i(\cdot,\cdot): \cX \times \hat{\cX}_i\rightarrow [0,\infty)$ be a distortion measure for $i$-th decoder.
The rates of code $(R_1,R_2)$ are simply defined as
\begin{align*}
R_1 =&\frac{1}{n}\log M_1\\
R_2 =&\frac{1}{n}\log M_1M_2.
\end{align*}
An $(n,R_1,R_2,D_1,D_2,\epsilon)$-successive refinement code is a coding scheme with block length $n$
and excess distortion probability $\epsilon$ where rates are $(R_1,R_2)$ and target distortions are $(D_1,D_2)$.
Since we have two decoders, the excess distortion probability is defined by $\Pr{d_i(X^n,\hat{X}_i^n)>D_i \mbox{ for some $i$}}$.
\begin{definition}
    A rate-distortion tuple $(R_1,R_2,D_1,D_2)$ is said to be achievable,
    if there is a family of $(n, R_1^{(n)}, R_2^{(n)}$, $D_1, D_2, \epsilon^{(n)})$-successive refinement code where
    \begin{align*}
        &\lim_{n\rightarrow\infty}R_i^{(n)}=R_i  \mbox{ for all $i$, }\\
        &\lim_{n\rightarrow\infty}\epsilon^{(n)}=0.
    \end{align*}
\end{definition}

For some special cases, both decoders can achieve the point-to-point optimum rates simultaneously.
\begin{definition}\label{def:successive refinability}
    Let $R_i(D_i)$ denote the rate-distortion function of the $i$-th decoder for $i\in\{1,2\}$.
    If the rate-distortion tuple $(R_1(D_1), R_2(D_2), D_1, D_2)$ is achievable,
    then we say the source is \emph{successively refinable at $(D_1,D_2)$}.
    If the source is successively refinable at $(D_1, D_2)$ for all $D_1,D_2$,
    then we say the source is \emph{successively refinable}.
\end{definition}

The following theorem provides a necessary and sufficient condition of successive refinable sources.
\begin{theorem}[\cite{equitz1991successive}, \cite{koshelev1980hierarchical}]\label{thm:iffcondition for successive refinability}
    A source is successively refinable at $(D_1,D_2)$ if and only if there exists
    a conditional distribution $P_{\hat{X}_1,\hat{X}_2|X}$ such that $X-\hat{X}_2-\hat{X}_1$ forms a Markov chain and
    \begin{align*}
        R_i(D_i) &= I(X;\hat{X}_i) \\
        \E{d_i(X,\hat{X}_i)}&\leq D_i
    \end{align*}
    for $i\in\{1, 2\}$.
\end{theorem}

The condition in the theorem holds for all discrete memoryless sources under logarithmic loss.
Note that the above results of successive refinability can easily be generalized to the case of $k$ decoders.
\section{One to One Correspondence Between General Distortion and Logarithmic Loss}\label{sec:Equivalence}

\subsection{Main Results}

Consider fixed-length lossy compression under an arbitrary distortion $d(\cdot, \cdot)$
as described in Section \ref{subsec:One Shot}.
We have a source $X$ with the finite alphabet $\cX=\{1,\ldots,r\}$, source distribution $P_X$,
and the finite reconstruction alphabet $\cXh=\{1,\ldots, s\}$.
For the fixed number of messages $M$, let $f^\star$ and $g^\star$ be the encoder and decoder
that achieve the optimum average distortion $D^\star(M)$.
I.e.,
\begin{align*}
    \E{d(X, g^\star(f^\star(X)))} = D^\star(M).
\end{align*}

Let $P_{\Xh|X}^\star$ denote the rate-distortion function achieving conditional distribution at distortion $D=D^\star(M)$.
I.e., $P_X \times P_{\hat{X}|X}^\star$ achieves the infimum in
\begin{align}
    R(D^\star(M)) = \inf_{P_{\hat{X}|X}:\E{d(X,\hat{X})}\leq D^\star(M)} I(X;\hat{X}). \label{eq:RDstar}
\end{align}
Note that $R(D^\star(M))$ may be strictly smaller than $\log M$ in general
since $R(\cdot)$ is an informational rate-distortion function
which does not characterize the best achievable performance for the ``one-shot'' setting in which $D^\star(M)$ is defined.

Now, we define the corresponding fixed-length lossy compresson problem under logarithmic loss.
In the corresponding problem, the source alphabet $\cX = \{1,\ldots, r\}$, source distribution $P_X$,
and the number of messages $M$ remain the same.
However, we have different reconstruction alphabet
$\cY = \{P_{X|\Xh}^{\star}(\cdot|\xh):\xh\in\cXh\}\subset \cM(\cX)$
where $P^\star$ pertains to the achiever of the infimum in \eqref{eq:RDstar}
associated with the original loss function.
Let the distortion now be logarithmic loss.

We can also connect the encoding and decoding schemes between the two problems.
Suppose $f:\cX\rightarrow \{1,\ldots,M\}$ and $g:\{1,\ldots,M\}\rightarrow \cXh$ are an encoder
and decoder pair in the original problem.
Let $f_{\ell} = f$ and $g_{\ell}:\{1,\ldots,M\}\rightarrow \cY$ where
\begin{align*}
    g_{\ell}(m) = P_{X|\Xh}^\star(\cdot|g(m)).
\end{align*}
Then, $f_{\ell}$ and $g_{\ell}$ are an valid encoder and decoder pair for the
corresponding fixed-length lossy compression problem under logarithmic loss.
Conversely, given $f_{\ell}$ and $g_{\ell}$, we can find corresponding $f$ and $g$ because
Lemma \ref{lem:uniqueness of rows} guarantees that $P_{X|\Xh}(\cdot|\xh)$ are distinct.

The following result shows the relation between the corresponding schemes.
\begin{theorem}\label{thm:correspondence}
    For any encoder-decoder pair $(f_{\ell},g_{\ell})$ for the corresponding fixed-length lossy compression problem
    under logarithmic loss, we have
    \begin{align*}
        &\E{\ell(X,g_{\ell}(f_{\ell}(X)))}\nonumber\\
        &= H(X|\Xh^{\star}) + \lambda^{\star}\left(\E{d(X,g(f(X)))}-D^{\star}(M)\right)\nonumber\\
        &\geq H(X|\Xh^{\star})
    \end{align*}
    where $(f,g)$ is the corresponding encoder-decoder pair
    for the original lossy compression problem.
    Note that $H(X|\Xh^{\star})$ and the expectations are with respect to
    the distribution $P_X \times P_{\Xh|X}^\star$.
    Moreover, equality holds if and only if
    $f_\ell = f^\star$ and $g_\ell(m) = P_{X|\Xh}^\star(\cdot|g^\star(m))$.
\end{theorem}

\begin{proof}
    We have
    \begin{align*}
        &\E{\ell(X,g_{\ell}(f_{\ell}(X)))}\nonumber\\
        & =\E{\ell\left(X,P_{X|\Xh}^{\star}(\cdot|g(f(X)))\right)}\\
        &= \E{\log\frac{1}{P_{X|\Xh}^{\star}(x|g(f(x)))}}.
    \end{align*}
    Then, \eqref{eq:PXgivenPXhstar} implies that
    \begin{align}
        &\E{\ell(X,g_{\ell}(f_{\ell}(X)))}\nonumber\\
        &=\E{\log\frac{1}{P_X(X)} - \jmath_X(X,D^\star(M))}\nonumber\\
        &~~~ + \E{\lambda^\star d(X,g(f(X))) - \lambda^\star D^\star(M)}\nonumber\\
        &= H(X|\Xh^{\star}) + \lambda^{\star}\left(\E{d(X,g(f(X)))}-D^{\star}(M)\right)
            \label{eq:expected tilted=RD}\\
        & \geq H(X|\Xh^{\star}) \label{eq:equalitycondition}
    \end{align}
    where \eqref{eq:expected tilted=RD} is because
    $\E{\jmath_X(X,D^\star(M))} = R(D^\star(M)) = I(X;\Xh^\star)$
    with respect to the distribution $P_X \times P_{\Xh|X}^\star$.
    The last inequality \eqref{eq:equalitycondition} is because $D^\star(M)$ is the minimum achievable
    average distortion with $M$ codewords.
    Equality holds if and only if $\E{d(X,g(f(X)))}=D^{\star}(M)$ which can be achieved by
    the optimum scheme for the original lossy compression problem.
    In other words, equality holds if
    \begin{align*}
        f^\star_\ell =& f^\star\\
        g^\star_\ell(m) =& P_{X|\Xh}^\star(\cdot|g^\star(m)).
    \end{align*}
\end{proof}

\begin{remark}
    In the corresponding fixed-length lossy compression problem under logarithmic loss,
    the minimum achievable average distortion given the number of codewords $M$ is
    \begin{align*}
        D^\star_{\ell}(M) = H(X|\Xh^\star)
    \end{align*}
    where the conditional entropy is with respect to the distribution $P_X\times P_{\Xh|X}^\star$.
\end{remark}

\subsection{Discussion}

\subsubsection{One-to-One Correspondence}\label{subsubsec:One-to-One Correspondence}
Theorem \ref{thm:correspondence} implies that for any fixed-length lossy compression problem,
we can find an equivalent problem under the logarithmic loss where optimum encoding schemes are the same.
Thus, without loss of generality, we can restrict our attention to the problem
under logarithmic loss with reconstruction alphabet
$\cY = \{q^{(1)},\ldots,q^{(s)}\}\subset\cM(\cX)$.

\subsubsection{Suboptimality of the Scheme}
Suppose $f$ and $g$ are sub-optimal encoder and decoder for the original fixed-length lossy compression problem,
then the theorem implies
\begin{align}
    &\E{\ell(X,g_\ell(X))} - H(X|\Xh^{\star})\nonumber\\
    &=  \lambda^{\star}\left(\E{d(X,g(f(X)))}-D^{\star}(M)\right)\label{eq:CostofSuboptimality}.
\end{align}
The left hand side of \eqref{eq:CostofSuboptimality} is the cost of sub-optimality
for the corresponding lossy compression problem.
On the other hand, the right hand side is proportional to the cost of sub-optimality
for the original problem.
In Section \ref{subsubsec:One-to-One Correspondence},
we discussed that the optimum schemes of the two problems coincide.
Equation \eqref{eq:CostofSuboptimality} shows stronger equivalence
that the cost of sub-optimalities are linearly related.
This simply implies a good code for one problem is also good for another.

\subsubsection{Operations on $\cXh$}
In general, the reconstruction alphabet $\cXh$ does not have any algebraic structure.
However, in the corresponding rate-distortion problem,
the reconstruction alphabet is the set of probability measures where we have
very natural operations such as convex combinations of elements, projection to a convex hull, etc.

\subsection{Exact Performance of Optimum Scheme}
In the previous section, we showed that there is a corresponding lossy compression problem
under logarithmic loss which shares the same optimum coding scheme.
In this section, we investigate the exact performance of the optimum scheme
for the fixed-length lossy compression problem under logarithmic loss,
where the reconstruction alphabet is the set of all measures on $\cX$, i.e., $\cM(\cX)$.
Although the optimum scheme associated with $\cM(\cX)$ may differ from
the optimum scheme with the restricted reconstruction alphabets $\cY$,
it may provide some insights.
Note that we are not allowing randomization, but restrict attention to deterministic schemes.

\subsubsection{Average Distortion Criterion}
In this section, we characterize the minimum average distortion $D^{\star}(M)$
when we have fixed number of messages $M$.
Let an encoder and a decoder be $f:\cX\rightarrow\{1,\ldots,M\}$
and $g:\{1,\ldots,M\}\rightarrow \cM(\cX)$ where $g(m) = q^{(m)}\in\cM(\cX)$.
Then, we have
\begin{align*}
    &\E{\ell(X,g(f(X)))}\nonumber\\
    &= \sum_{x\in\cX} P_X(x) \log \frac{1}{q^{(f(x))}(x)}\\
    &= H(X)+\sum_{m=1}^M \sum_{x\in f^{-1}(m)} P_X(x) \log \frac{P_X(x)}{q^{(m)}(x)}\\
    &= H(X)+\sum_{m=1}^M u_m \log u_m\nonumber\\
    &~~~+\sum_{m=1}^M  u_m\sum_{x\in f^{-1}(m)} \frac{P_X(x)}{u_m} \log \frac{P_X(x)/u_m}{q^{(m)}(x)}
\end{align*}
where $u_m = \sum_{x\in f^{-1}(m)} P_X(x)$.
Since $P_{X|f(X)}(x|m) = \frac{P_X(x)}{u_m}$ for all $x\in f^{-1}(m)$, we have
\begin{align*}
    &\E{\ell(X,g(f(X)))}\nonumber\\
    &= H(X)-H(f(X))\nonumber\\
    &~~~+\sum_{m=1}^M  u_m D\left(P_{X|f(X)}(\cdot|m)\dbar q^{(m)}\right)\\
    &\geq H(X)-H(f(X)).
\end{align*}
Equality can be achieved by choosing $q^{(m)} = P_{X|f(X)}(\cdot|m)$
which can be done no matter what $f$ is.
Thus, we have
\begin{align*}
    D^{\star}(M) = H(X) - \max_{f: |f|\leq M} H(f(X)).
\end{align*}
Note that one trivial lower bound is
\begin{align*}
    D^{\star}(M) \geq H(X) - \log M.
\end{align*}

\subsubsection{Excess Distortion Criterion}
In this section, we characterize the minimum number of codewords $M^\star(D,\epsilon)$
that can achieve the distortion $D$ and the excess distortion probability $\epsilon$.
Let an encoder and a decoder be $f:\cX\rightarrow\{1,\ldots,M\}$
and $g:\{1,\ldots,M\}\rightarrow \cM(\cX)$ where $g(m) = q^{(m)}\in\cM(\cX)$.
Since $\ell(x,q)\leq D$ is equivalent to $q(x)\geq e^{-D}$, we have
\begin{align*}
    1-P_e =& \sum_{x\in\cX} P_X(x) \1\left(q^{(f(x))}(x)\geq e^{-D}\right)\\
    =&\sum_{m=1}^M\sum_{x\in f^{-1}(m)} P_X(x) \1\left(q^{(m)}(x)\geq e^{-D}\right).
\end{align*}

However, at most $\lfloor e^D\rfloor$ of the $q^{(m)}(x)$ can be larger than $e^{-D}$
where $\lfloor x\rfloor$ is the largest integer that is smaller than or equal to $x$.
Thus, we can cover at most $M\cdot \lfloor e^D\rfloor$ of the source symbols
with $M$ codewords. Suppose $P_X(1)\geq P_X(2)\geq \cdots \geq P_X(r)$,
then the optimum scheme is
\begin{align*}
    f(x) =& \left \lceil\frac{x}{\lfloor e^D \rfloor}\right\rceil\\
    q^{(m)}(x) =& \begin{cases}
        1/\lfloor e^D \rfloor &\mbox{if $f(x)=m$}\\
        0&\mbox{otherwise}
    \end{cases}
\end{align*}
where $q^{(m)} = g(m)$.
The idea is that each reconstruction symbol $q^{(m)}$ covers $\lfloor e^D\rfloor$ number of source symbols
by assigning probability mass $1/\lfloor e^D\rfloor$ to each of them.

The above optimum scheme satisfies
\begin{align*}
    1-P_e =& \sum_{x=1}^{M\cdot \lfloor e^D \rfloor} P_X(x)\\
    =& F_X\left(M\cdot \lfloor e^D \rfloor\right)
\end{align*}
where $F_X(\cdot)$ is the cumulative distribution function of $X$.
This implies that the minimum error probability is
\begin{align*}
    \epsilon^{\star}(M,D) = 1-F_X\left(M\cdot \lfloor e^D \rfloor\right).
\end{align*}

On the other hand, if we fix the target error probability $\epsilon$,
the minimum number of codewords is
\begin{align*}
    M^{\star}(D,\epsilon) = \left\lceil\frac{F_X^{-1}(1-\epsilon)}{\lfloor e^D \rfloor}\right\rceil
\end{align*}
where $F_X^{-1}(y) = \argmin_{1\leq x\leq r} \{x: F_X(x)\geq y\}$
and $\lceil x\rceil$ is the smallest integer that is larger than or equal to $x$.
Note that if we allow variable length coding without prefix condition,
the optimum coding scheme is similar to the optimum nonasymptotic lossless coding introduced in \cite{kontoyiannis2014optimal}.

\section{Successive Refinability}\label{sec:Successive Refinability}
We considered the fixed-length lossy compression problem so far.
In this section, we provide another universal property of logarithmic loss
where the source is discrete memoryless.

\subsection{Main Results}

Consider the successive refinement problem with a discrete memoryless source as described in Section \ref{subsec:Successive Refinability}.
Specifically, we are instrested in the case where the first decoder is under logarithmic loss
and the second decoder is under some arbitrary distortion measure $d$.
Using the result from previous section, there is an equivalent rate-distortion problem
under logarithmic loss for the second decoder.
Since any discrete memoryless source under logarithmic loss is successively refinable,
one may argue that the source is successively refinable under this setting.
However, this can be misleading since we cannot directly apply our result to
discrete memoryless source.
This is mainly because the decoder only considers product measures
when the source is discrete memoryless.
Moreover, the equivalent rate-distortion problem under logarithmic loss has
restricted reconstruction set with only finitely many measures,
while successive refinability of logarithmic loss is guaranteed
when the reconstruction alphabet is the set of all measures on $\cX$.

Despite of these misconceptions, we show that our initial guess was correct, i.e., it is successively refinable.
This provides an additional universal property of logarithmic loss
in the context of the successive refinement problem.

\begin{theorem}\label{thm:successive refinability of log loss}
    Let the source be arbitrary discrete memoryless.
    Suppose the distortion criterion of the first decoder is logarithmic loss
    while that of the second decoder is an arbitrary distortion criterion
    $d:\cX\times\hat{\cX}\ra[0,\infty]$.
    Then the source is successively refinable.
\end{theorem}
\begin{proof}
    The source is successively refinable at $(D_1,D_2)$ if and only if
    there exists a $X-\Xh-Q$ such that
    \begin{align*}
        I(X;Q) =& R_1(D_1), \quad \E{\ell(X,Q)}\leq D_1\\
        I(X;\hat{X}) =& R_2(D_2), \quad \E{d(X,\hat{X})}\leq D_2.
    \end{align*}

    Let $P_{\Xh^\star|X}$ be the informational rate-distortion function achieving
    conditional distribution for the second decoder.
    I.e.,
    \begin{align*}
        I(X;\Xh^\star) = R_2(D_2), \quad \E{d(X,\Xh^\star)}=D_2.
    \end{align*}
    Consider a random variable $Z\in\cZ$ such that the joint distribution of
    $X,\Xh^\star,Z$ is given by
    \begin{align*}
        P_{X,\Xh,Z}(x,\xh,z) = P_{X,\Xh}(x,\xh)P_{Z|\Xh}(z|\xh)
    \end{align*}
    and $H(X|Z) = D_1$.
    We assume that all the $P_{X|Z}(\cdot|z)$ are distinct for all $z\in\cZ$.
    If we let $Q =  P_{X|Z}(\cdot|Z)$ and $q^{(z)} = P_{X|Z}(\cdot|z)$ for all $z\in\cZ$,
    then $X-\Xh^\star-Q$ forms a Markov chain and
    \begin{align*}
        P_{X|Q}(x| q^{(z)}) = q^{(z)}(x).
    \end{align*}
    
    Since $Q = P_{X|Z}(\cdot|Z)$ is a one-to-one mapping, we have
    \begin{align*}
        I(X;Q) =& I(X;Z)= H(X) -D_1= R_1(D_1).
    \end{align*}
    Also, we have
    \begin{align*}
        \E{\ell(X,Q)} =& \E{\log\frac{1}{P_{X|Z}(X|Z)}}= H(X|Z)= D_1.
    \end{align*}
    
    We have no constraints on the set $\cZ$ and the only requirements for
    the random variable $Z$ is $H(X|Z) = D_1$.
    Therefore, we can always find such random variable $Z$,
    and we can say that the source and respective distortion measures are successively refinable.
\end{proof}

The key idea of the theorem is that \eqref{eq:iff condition} is the only loose
required condition for the rate-distortion function achieving conditional distribution.
Thus, for any distortion criterion in the second stage,
we are able to choose an appropriate distribution $P_{X,\Xh,Q}$
that satisfies both \eqref{eq:iff condition} and the condition for successive refinability.

\begin{remark}
    We would like to point out that the way of constructing the joint
    distribution $P_{X,\Xh,Q}$ in the proof using random variable $Z$
    is the only possible way.
    More precisely, consider a Markov chain $X-\Xh-Q$ that satisfies the condition
    for successive refinability, then there exists a random variable $Z$
    such that $Q = P_{X|Z}(\cdot|Z)$ where $X-\Xh-Z$ forms a Markov chain.
    This is because we can have $Z=Q$, in which case $P_{X|Z}(\cdot|Z) =  P_{X|Q}(\cdot|Q)=Q$. 
\end{remark}

Theorem \ref{thm:successive refinability of log loss} can be generalized to
successive refinement problem with $K$ intermediate decoders.
Consider random variables $Z_k\in\cZ_k$ for $1\leq k\leq K$ such that
$X-\Xh-Z_K-\cdots-Z_1$ forms a Markov chain and the joint distribution of
$X,\Xh^\star,Z_1,\ldots,Z_K$ is given by
\begin{align*}
    &P_{X,\Xh,Z_1,\ldots,Z_K}(x,\xh,z_1,\ldots,z_K)\nonumber\\ 
    &= P_{X,\Xh}(x,\xh)P_{Z_1|\Xh}(z_1|\xh) \prod_{k=1}^{K-1} P_{Z_{k+1}|Z_k}(z_{k+1}|z_k)
\end{align*}
where $H(X|Z_k) = D_k$ and all the $P_{X|Z_k}(\cdot|z_k)$ are distinct for all $z_k\in\cZ_k$.
Similar to the proof of Theorem \ref{thm:successive refinability of log loss},
we can show that $Q_k =  P_{X|Z_k}(\cdot|Z_k)$ for all $1\leq k\leq K$
satisfy the condition for successive refinability.
Thus, we can conclude that any discrete memoryless source with $K$ intermediate decoders
is successively refinable as long as all the intermediate decoders are under logarithmic loss.

\section{Conclusion}
To conclude our discussion, we summarize our main contributions.
We showed that for any fixed length lossy compression problem under arbitrary distortion measure,
there exists a corresponding lossy compression problem under logarithmic loss
where the optimum schemes coincide.
Moreover, we proved that a good scheme for one lossy compression problem is also good for the corresponding problem.
This provides an algebraic structure on any reconstruction alphabet.
On the other hand, in the context of successive refinement problem,
we showed another universal property of logarithmic loss
that any discrete memoryless source is successively refinable
as long as the intermediate decoders operate under logarithmic loss.


\bibliographystyle{IEEEtran}
\bibliography{IEEEabrv,AlbertRef}

\end{document}